\documentclass[runningheads]{llncs}
\usepackage[T1]{fontenc}
\usepackage{graphicx}
\usepackage{amsmath}
\usepackage{amssymb}
\usepackage{xspace}
\usepackage{algorithm}
\usepackage{algpseudocode}
\usepackage{xparse}
\usepackage[margin=1.1in]{geometry}
\usepackage{marvosym}

\usepackage{auxlib}

\begin{document}
\title{Mechanism Design for Exchange Markets\thanks{
  This work is supported by 
  Wuhan East Lake High-Tech Development Zone (also known as the Optics Valley of China, or OVC) National Comprehensive Experimental Base for Governance of Intelligent Society, 
  NSFCs(No. 12471339, No. 62302166 and No. 62172012), 
  and the Key Laboratory of Interdisciplinary Research of Computation and Economics (SUFE), Ministry of Education.
}}

\author{
Yusen Zheng
\inst{1,2}\textsuperscript{(\Letter)}
\and
Yukun Cheng
\inst{3}\textsuperscript{(\Letter)}
\and
Chenyang Xu
\inst{4}\textsuperscript{(\Letter)}
\and
Xiaotie Deng
\inst{1,2}\textsuperscript{(\Letter)}
}

\authorrunning{Yusen Zheng, Yukun Cheng, Chenyang Xu, and Xiaotie Deng}

\institute{
CFCS, School of Computer Science, Peking University, Beijing 100871, China
\and
PKU-WUHAN Institute for Artificial Intelligence, Wuhan 430075, China
\and
School of Business, Jiangnan University, Wuxi 214122, China
\and
Software Engineering Institute, East China Normal University, China\\
\email{yusen@stu.pku.edu.cn, ykcheng@amss.ac.cn, cyxu@sei.ecnu.edu.cn, xiaotie@pku.edu.cn}
}
\maketitle              %

\begin{abstract}
Exchange markets are a significant type of market economy, in which each agent holds a budget and certain (divisible) resources available for trading. Most research on equilibrium in exchange economies is based on an environment of completely free competition. However, the orderly operation of markets in reality also relies on effective economic regulatory mechanisms. This paper initiates the study of the mechanism design problem in exchange markets, exploring the potential to establish truthful market rules and mechanisms. This task poses a significant challenge as unlike auctioneers in auction design, the mechanism designer in exchange markets lacks centralized authority to fully control the allocation of resources.
  
In this paper, the mechanism design problem is formalized as a two-stage game. In stage 1, agents submit their private information to the manager, who then formulates market trading rules based on the submitted information. In stage 2, agents are free to engage in transactions within these rules, ultimately reaching an equilibrium. We generalize the concept of liquid welfare from classical budget-feasible auctions and use \emph{market liquid welfare} as a measure to evaluate the performance of the designed mechanism. Moreover, an extra concept called \emph{profitability} is introduced to assess whether the market is money-making (profitable) or money-losing (unprofitable). Our goal is to design a truthful mechanism that achieves an (approximate) optimal welfare while minimizing unprofitability as much as possible. Two mechanisms for the problem are proposed. The first one guarantees truthfulness and profitability while approaching an approximation ratio of $1/2$ in large markets. The second one is also truthful and achieves $1/2$ approximation in general markets but incurs bounded unprofitability. Our aim is for both mechanisms to provide valuable insights into the truthful market design problem.


\end{abstract}

\section{Introduction}

Today we are witnessing a renewed interest about models for exchanging services and resources in markets. The agents in the market endow with a certain amount of money and resources for trading, Any two agents, as long as they are willing, can participate in transactions where one purchases resources from the other using money.
From the traditional stock trading in securities markets to the exchange of data\cite{georgiadis2015exchange}, bandwidth \cite{wu2007proportional}, and computational power \cite{misra2010incentivizing} in information networks, and onto the securitization of NFTs in blockchain systems \cite{chen2022absnft}, all bear witness to the powerful role of the exchange economy.
Numerous related works have focused on analyzing the existence and the computation of equilibria \cite{arrow1959stability,arrow1958stability,deng2002complexity}, which are typically achieved in
a fully competitive environment. However, in reality, markets are not completely free and require regulation and governance through certain economic control measures by market managers.
Thus this paper adopts a fresh perspective in examining the exchange market by looking at it through the lens of a market manager, investigating the potential to establish truthful market rules and mechanisms to promote market prosperity.

Truthful mechanism design is one of the fundamental research topics in game theory. There is a vast body of literature in this field, covering various game scenarios such as auction design~\cite{myerson1981optimal,vickrey1961counterspeculation,clarke1971multipart,groves1973incentives}, scheduling design~\cite{christodoulou2009mechanism,chen2016efficient,heydenreich2007games}, and matching design~\cite{chen2011mechanism,budish2012matching}. In most of these scenarios, mechanism designers are typically considered to have the authority to allocate resources and manage finances. For example, in auction design, the auctioneer has the power to determine which bidders receive the items and how much money they are charged for them.

In contrast, such centralized authority is lacking in a free exchange market. As market managers, we may have the ability to set the prices of goods and define the maximum and minimum transaction limits for each agent. For instance, in the regulation of securities trading markets, mechanisms such as "circuit breakers" or "position limits" are commonly used to control the trading volume of participants. Additionally, these markets offer investors timely stock price information via "real-time quotes." But, we cannot enforce the allocation of goods for each agent. This limitation poses a significant challenge in characterizing the market design model and developing mechanisms that ensure truthfulness.

\subsection{Our Contributions}

In this paper, we formalize the market design problem by precisely defining the powers of the mechanism designer and the social objectives. As illustrated in~\cref{fig:framework}, we model the resource exchange process in the exchange market as a two-stage game.
In the first stage, each agent, acting as the leader, reports their private information to the market manager. Using the reported information, the market manager employs a mechanism to establish the transaction limits for each agent, which represent the maximum number of resource units an agent is permitted to buy or sell. Additionally, the market manager sets the distinct price $\lambda_i$ that an agent must pay or will receive for each unit of resource exchange.

\begin{figure}[ht]
    \centering
    \includegraphics[width=0.6\linewidth]{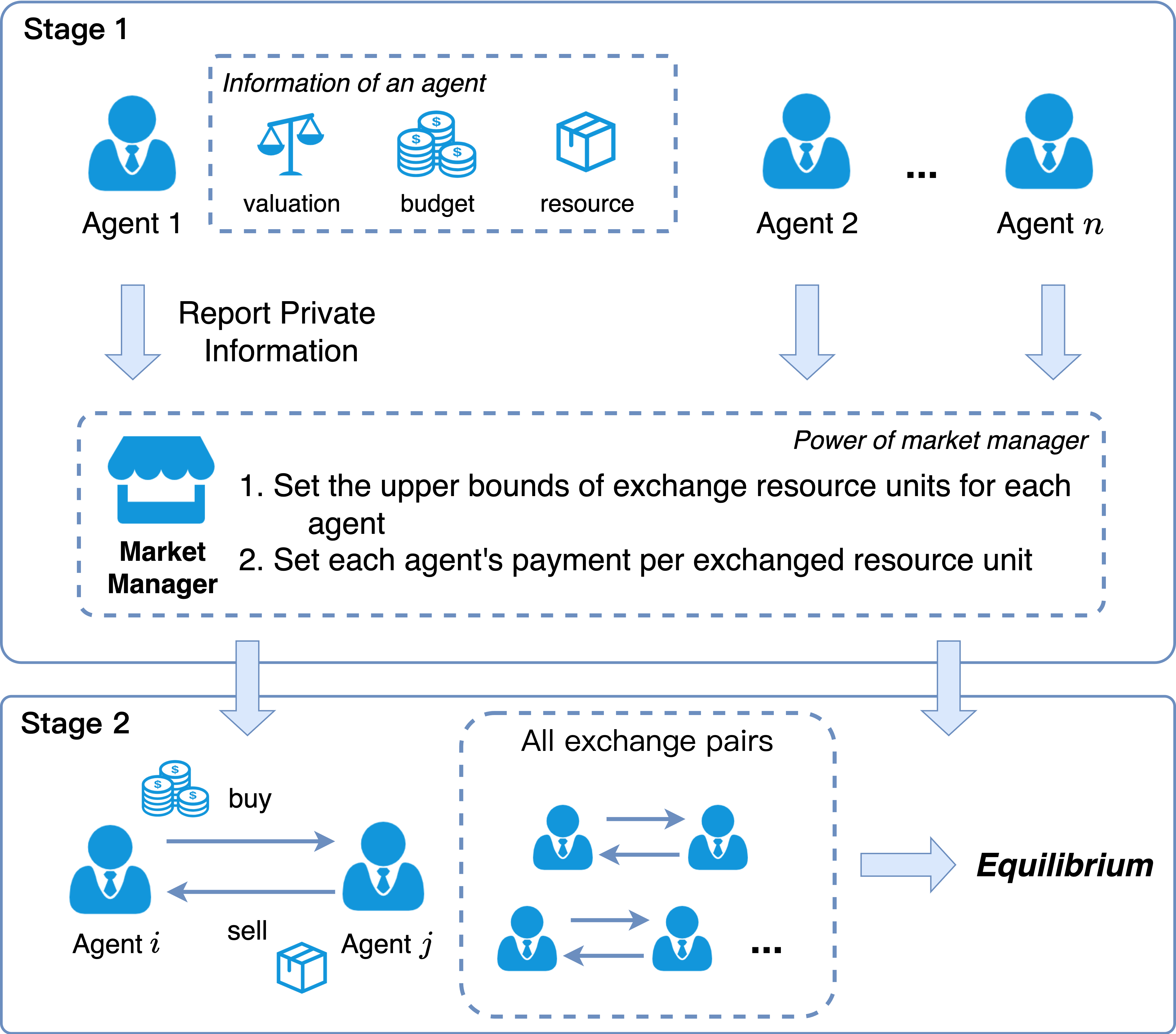}
    \caption{The 2-stage game framework in the exchange market.}
    \label{fig:framework}
\end{figure}

In the second stage, each agent acts selfishly under the constraints set in the first stage. To be specific, an agent $i$ can initiate a resource exchange with any other agent $j$. After the exchange, one agent gains resources and makes a monetary payment. Conversely, the other agent loses resources and receives money based on her price. It is crucial to note that the payment discrepancy between the two agents might not be zero. This discrepancy could be viewed as the market applying taxes or offering subsidies on the transactions.
We are not particularly concerned with the buying or selling units in each exchange, but rather focus on the eventual 'net' trading unit $x_i$ of each agent. %
Therefore, we can use the trading unit profile ${\bf x}=\{x_i\}_{i}$ to characterize a state. The total valuation and payment for each agent can be calculated by multiplying $x_i$ by their value $v_i$ and price $\lambda_i$, respectively. Due to their inherent selfishness, each agent strives to maximize their utility by engaging actively in resource exchanges. At the same time, the market eventually reach an equilibrium, where no agent is motivated to engage in further purchases or sales.

Lu et al. \cite{lu2017liquid} noted that achieving any reasonable guarantee for the social welfare objective is impossible, even in the simplest setting of a single-item budget-feasible auction. Given that the budget-feasible auction is a special case of the exchange market where the auctioneer is the unique seller, this observation motivates us to explore the concept of liquid welfare within the exchange market. Unlike budget-feasible auctions, where liquid welfare hinges solely on buyers, in an exchange market, each agent can act both as a seller and a buyer. Consequently, the calculation of social welfare must encompass the contributions from both sellers and buyers. To address this, we introduce the concept of {\it market liquid welfare} for the exchange market. Essentially, each agent's contribution to the market liquid welfare comprises the total valuation of all their original resources plus the valuation for the allocated bundle capped by their budgets.
Overall, our contributions are summarized as follows:
\begin{itemize}[left=1em]
    \item We propose a model to address the mechanism design problem in exchange markets. In this model, agents in the market first report their private information, after which the market manager formulates market trading rules. Subsequently, agents engage in transactions freely within these rules. Each agent can be viewed as performing a bi-level optimization (\cref{def:market_utility}), where reporting different private information leads to changes in the rules and the overall market supply-demand relationship.
    
    On the other hand, we generalize the concept of liquid welfare from classical budget-feasible auctions and propose market liquid welfare (\cref{def:MLW}) as a measure to evaluate the performance of market rules. Additionally, considering the real-world attribute of markets where they can either generate profits or incur losses, we introduce the concept of profitability (\cref{def:profit}) to characterize this aspect.
    \vspace{0.5em}
    \item We characterize the optimal solution to the market design problem and demonstrate that there always exists a profitable (uniform price) mechanism that achieves the optimal market liquid welfare (\cref{lem:mop}). However, this mechanism suffers from a truthfulness issue. Therefore, we refine the mechanism and propose a truthful uniform price mechanism (\cref{thm:large_market}). Under a large market assumption, the mechanism's approximation ratio approaches $1/2$. This result can further be extended to the multi-parameter setting where budgets and initially held resources are private (\cref{thm:uniform_large_mp}).
    \vspace{0.5em}
    \item Without the large market assumption, maintaining truthfulness, profitability, and a good approximation ratio presents a significant challenge. To address this, we propose a differential pricing mechanism that ensures truthfulness and achieves a 1/2-approximation ratio (\cref{thm:mechs}). While this mechanism may lead to a money-losing market, we can prove that the unprofitability of the mechanism is bounded.
    The basic algorithmic idea is similar to the budget-feasible mechanism proposed in~\cite{dobzinski2014efficiency}, but it is important to note that the "bi-level" nature of our model and the dynamic nature of market supply and demand pose new challenges for analyzing truthfulness and approximation ratios.
\end{itemize}

\subsection{Related Work}
\subsubsection*{Exchange Markets.}  For the exchange market, much of the research has traditionally focused on market equilibrium, particularly on its existence and computation. Arrow and Debreu \cite{arrow1954existence} pioneered this area by proving the existence of equilibrium in an Arrow–Debreu market, assuming concave utility functions. Subsequent studies, such as those referenced by Deng et al. \cite{deng2002complexity}, Duan and Mehlhorn \cite{duan2016improved}, Jain et al. \cite{jain2007polynomial}, and Ye \cite{ye2008path}, have primarily employed combinatorial methods to tackle the computational challenges of determining market equilibrium.
In a different vein, Hurwicz \cite{hurwicz1986informational} explored the incentive structures within market equilibrium, revealing the absence of any incentive-compatible mechanism that ensures market equilibrium. Addressing this limitation, Roberts and Postlewaite \cite{roberts1976incentives} developed a theory that supports the price-taking behavior in large exchange economies. They demonstrated that agents in an infinitely large market behave nearly truthfully. Jackson \cite{jackson1992incentive} further bolstered this theory by proving that the collective demand function of strategic (i.e., cheating) agents converges to the total demand of price-taking agents as the population size increases.
More recently, Cheng et al. have provided the first mathematical proof demonstrating that the market equilibrium mechanism remains truthful against manipulative strategies \cite{cheng2023truthfulness} and explore the incentive ratio of the market equilibrium mechanism against Sybil attack \cite{cheng2022tight}.

\subsubsection*{Budget-Feasible Auctions.}
In scenarios with budget constraints, the concept of Pareto Efficiency often serves as a foundation for assessing outcomes in economic models and is widely discussed in literature \cite{dobzinski2012multi,fiat2011single}. Nonetheless, Pareto Efficiency, being a binary criterion, does not facilitate the development of approximation algorithms since it solely indicates whether any possible improvement exists without quantifying how close an outcome is to optimal.
Recognizing this limitation, Dobzinski et al. introduced a more quantifiable welfare metric known as liquid welfare \cite{dobzinski2014efficiency}. This metric innovatively incorporates both the agents' valuations of items and their purchasing powe. They further leveraged this metric to design a mechanism that achieves a 2-approximation ratio for auctions involving public budgets. Building on this approach, Lu et al. improved this approximation ratio to $\frac{1+\sqrt{5}}{2}$, demonstrating a more efficient mechanism for handling such complex auction scenarios \cite{lu2015improved}.
Further advancements were made by Lu et al. who applied the concept of liquid welfare to multi-item auctions with budgets, exploring its implications and utility in more diversified auction frameworks \cite{lu2017liquid}. 

\subsubsection*{Large Market Assumption.}
The large market assumption is a widely-used assumption in budget-feasible auction design. It was first formally described in~\cite{anari2014mechanism}. 
Intuitively, a large market refers to a market with a large number of participants, where each participant's contribution to the total market is relatively small. 
By leveraging this assumption, the authors presented a mechanism for crowdsourcing with a competitive ratio of $1 - 1/e$, surpassing the theoretical upper limit of $0.5$ for the competitive ratio of randomized mechanisms without the large market assumption. Under the large market assumption, Lu and Xiao~\cite{lu2017liquid} proposed a truthful budget-feasible mechanism with a constant approximation ratio for multi-item auctions.

\subsection{Paper Organization}
In~\cref{sec:model}, we formally introduce our model. Subsequently, in~\cref{sec:mop}, we characterize the structure of the model's optimal solutions. \cref{sec:mp} presents a uniform price mechanism under the large market assumption, while \cref{sec:sp} addresses a more general setting and proposes a truthful mechanism. The paper concludes in~\cref{sec:con}.

\section{Model}\label{sec:model}
This section provides a formal definition of the market design model. An exchange market consists of $n$ agents,  where each agent $i$ possesses $\Gamma_i$ units of (divisible) resources available for trading. Each agent $i$ is assigned a budget $B_i$, representing the amount of money she can allocate for purchasing resources, and a (private) valuation $v_i$ for one unit of the resource.

As shown in~\cref{fig:framework}, there are two stages. In the first stage, each agent reports their private valuation $b_i$ to the market manager. Then, based on the reported valuations $\{b_i\}_{i\in [n]}$ and the publicly available information $\{B_i, \Gamma_i\}_{i\in [n]}$, the market manager formulates the market rules, defining the exchange constraints for each agent. In the second stage, under the satisfaction of the exchange constraints, agents selfishly engage in resource exchanges with each other. 

\subsubsection*{Exchange Constraints.} Denote by $x_i$ the \emph{net}\footnote{It is possible for an agent to initially purchase $a$ units and subsequently sell $b$ units. However, we only consider the net change: $x_i=a-b$. } amount of resources agent $i$ ultimately purchases in the market and by $p_i$ the \emph{net} amount of money agent~$i$ ultimately spends (if $x_i,p_i<0$, it indicates that the agent sells $|x_i|$ resource units and receives $|p_i|$ money). The exchange constraints for each agent can be represented by an interval $I_i=[-l_i,r_i]$ with $l_i,r_i\geq 0$, which restricts the range of $x_i$, and a value $\lambda_i\geq 0$, representing the unit price at which the agent buys or sells resources, i.e., $p_i = \lambda_i \cdot x_i$. 

To illustrate, suppose the market sets the constraint $(I_i = [-50, 100], \lambda_i = 10)$ for agent $i$. This means that the agent is allowed to sell at most 50 units of resources ``net'', and purchase up to 100 units of resources ``net''. For every resource unit the agent buys (from any other agent), she needs to pay \$10; Conversely, for every resource unit she sells (to any other agent), she will receive \$10. 

Note that the market does not limit the number of times trading in transactions agents can engage in, as long as they finally satisfy the constraints. When agent $i$ purchases $x$ units of resources from agent $j$, agent $i$ needs to pay $\lambda_i \cdot x$ while agent $j$ will receive $\lambda_j \cdot x$ as payment.
If $\lambda_i > \lambda_j$, the payment difference can be seen as a \emph{transaction tax} imposed by the system. Conversely, if $\lambda_i < \lambda_j$, the difference can be viewed as a \emph{subsidy} provided by the system for that transaction\footnote{This is a common strategy~\cite{pigou2017economics} to activate a market.}

Let $\vA=\{b_i,B_i,\Gamma_i\}_{i\in [n]}$ be a set of agent parameters. A market mechanism is a function $\cM$ that maps $\vA$ to a set of exchange constraints $\{(I_i, \lambda_i)\}_{i\in [n]}$.

\subsubsection*{Agent Utility.} Each agent uses a constrained quasi-linear utility. When agent~$i$ buys $x_i$ resource units and pays $p_i$, she gets the utility
$$
u_i (x_i,p_i):= \left\{
\begin{aligned}
&v_i \cdot x_i - p_i \;\;\;\; & \text{if $p_i\leq B_i$ and $x_i \geq -\Gamma_i$ ,}   \\
&-\infty\;\;\;\; & \text{otherwise.} \\
\end{aligned}
\right.
$$
In this work, we assume that agents are \emph{resource-preferable}, having a preference for acquiring resources rather than holding money. This means that when $u_i(x_i, p_i)=u_i(x_i',p_i')$ and $x_i <x_i'$, the agent prefers $(x_i',p_i')$ as it provides her with more units of resources.

Clearly, given exchange constraints $(\vI,\vlam)=\{ (I_i,\lambda_i) \}_{i\in [n]}$, an agent with $v_i \geq \lambda_i$ becomes a buyer, aiming to purchase as many resources as possible. On the other hand, an agent with $v_i < \lambda_i$ becomes a seller, seeking to sell as many resources as possible. Thus, when agents act selfishly, the exchange market must eventually reach an equilibrium market state $(\vx,\vp)=\{(x_i,p_i)\}_{i\in [n]}$ where either no buyer is willing to buy or no seller has remaining resources to sell. More formally, we define the following:

\begin{definition}[Reachability]\label{def:reach}
     Given an agent set $\vA=\{(v_i,B_i,\Gamma_i)\}_{i\in [n]}$ and exchange constraints $(\vI,\vlam)$, a market state $(\vx,\vp)=\{(x_i,p_i)\}_{i\in [n]}$ is \emph{reachable} to $(\vI,\vlam)$ if the following conditions are satisfied: 
    \begin{enumerate}[left=3em]
        \item[(1)] $\vx$ is self-consistent: $\sum_{i\in [n]}x_i = 0$.
        \item[(2)] $(\vx,\vp)$ is feasible: $\forall i\in [n]$, $x_i\in I_i$ and $p_i=\lambda_i \cdot x_i$.
        \item[(3)] $(\vx,\vp)$ is an equilibrium: either no buyer is willing to buy ($ p_i = B_i$ for any buyer $i$ with $v_i\geq \lambda_i$), or no seller has remaining resources to sell ($x_j = -\Gamma_j$ for any seller $j$ with $v_i<\lambda_j$).
    \end{enumerate}
     
\end{definition}

There could be more than one reachable market state $(\vx,\vp)$, depending on the order in which exchanges occur within the market. Consider a market mechanism $\cM$. We use $\spi(\cM,\vb, \vA )$ to denote the set of reachable states when agents in $\vA$ report a valuation profile $\vb=\{b_i\}_{i\in [n]}$ and $\cM$ is used to establish exchange constraints for them.

Since agents cannot foresee which equilibrium state will be reached in stage~2 when they report valuations $\vb$, we employ the worst-case analysis framework, and assume that agents are conservative, calculating their bidding strategies based on the minimum utility they can obtain.

\begin{definition}[Agent Market Utility]\label{def:market_utility}
    Given a mechanism $\cM$, an agent set~$\vA$, and a reported valuation profile $\vb$, the market utility of agent $i$ is defined as her minimum utility among all reachable states: 
    \[ U_i = \min_{(\vx,\vp)\in \spi(\cM,\vb, \vA)} u_i(\vx,\vp). \]
\end{definition}

\begin{definition}[Truthful Market]\label{def:truthful}
    A market mechanism $\cM$ is considered \emph{truthful} if for each agent, reporting the private valuation truthfully always maximizes her market utility, regardless of the reports submitted by other agents.
\end{definition}

We further use the worst-case framework to evaluate the performance of a mechanism.

\begin{definition}[Profitability]\label{def:profit}
   A market mechanism $\cM$ is considered \emph{profitable} if for any agent set $\vA$ with reported valuations $\vb$ and any reachable state $ (\vx,\vp)\in \spi(\cM,\vb, \vA)$, we have  $\sum_{i\in [n]} p_i \geq 0$; otherwise, it is referred to as a \emph{unprofitable} mechanism.
\end{definition}

The profitability metric assesses whether a mechanism results in a money-making or a money-losing market. If the mechanism is unprofitable, it means that on average, the market manager needs to subsidize each transaction. 

\begin{definition}[Market Liquid Welfare]\label{def:MLW} 
   Given a mechanism $\cM $ and an agent set $\vA=\{(v_i,B_i,\Gamma_i)\}_{i\in [n]}$ with reported valuations $\vb$, the market liquid welfare of a reachable state $(\vx,\vp)\in \spi(\cM,\vb,\vA)$ is defined as 
   \[ \mlw(\vx) = \sum_{i\in [n]} v_i\cdot \Gamma_i + \min\{v_i\cdot x_i, B_i\}. \]
   Further, let the market liquid welfare returned by mechanism $\cM$ be the worst reachable $\mlw(\vx)$: 
   \[ \mlw(\cM,\vb, \vA) = \min_{(\vx,\vp)\in \spi(\cM,\vb,\vA)} \mlw(\vx). \]
\end{definition}

Optimizing the market liquid welfare ensures that resources flow towards the agents with high valuations while also preventing those agents with limited budgets from obtaining excessive resources.
Our goal is to design a truthful and profitable mechanism that maximizes the market liquid welfare.

\section{Market Optimal Price}\label{sec:mop}

Compared to classical auction design, profitability poses a new challenge in market design. It can be observed that if a mechanism sets a uniform resource price for all agents, it is guaranteed to be profitable as no transaction needs to be subsidized.
Refer to such mechanisms as \emph{uniform-price}.
To gain algorithmic insight into market design, we start by temporarily setting aside the issue of truthfulness and exploring the best market liquid welfare achievable by uniform-price mechanisms.

\begin{lemma}[Market Optimal Price]\label{lem:mop}
    Given any agent set $\vA$, there always exists an exchange interval $I_i^*$ for each agent and a uniform \emph{market optimal price} (MOP) $ \lambda^* $ such that the optimal market liquid welfare can always be obtained when each agent $i$ subjects to constraint $(I_i^*, \lambda^*)$.
\end{lemma}

To streamline the proof, we begin by claiming that when the given $\{(I_i,\lambda_i)\}_{i\in[n]}$ satisfies an \emph{equilibrium-unique property}, the behaviors of selfish agents become controllable, resulting in a bounded market liquid welfare. Due to space limitations, we defer the proof to~\cref{sec:leem:eq_pro}.

\begin{lemma}[Equilibrium-Unique Property]\label{lem:equilibrium_property}
    Given a set of agents $\vA= \{(v_i,B_i,\Gamma_i)\}_{i\in [n]}$ and exchange constraints $\{(I_i=[-l_i,r_i],\lambda_i)\}_{i\in[n]}$,
    we refer to the set of agents with $v_i\geq \lambda_i$ as the buyer set $\cB$ and the set of other agents as the seller set $\cS$. If  
     \[ \sum_{i\in \cB} \min\left\{r_i, \frac{B_i}{\lambda_i}\right\} = \sum_{j\in \cS} \min\left\{l_j, \Gamma_j\right\}  \]
    then the reachable state $(\vx,\vp)$ is unique, with $x_i= \min\left\{r_i, \frac{B_i}{\lambda_i}\right\} $ $\forall i\in \cB$  and $x_j = - \min\left\{l_j, \Gamma_j\right\}$ $\forall j\in \cS$. Further, the market liquid welfare is 
    \[ \sum_{i\in \cB} \left( v_i\cdot \Gamma_i + \min\{v_i\cdot r_i, B_i\}   \right) + \sum_{j\in \cS} v_j \cdot \max\{0, \Gamma_j - l_j\}. \]

\end{lemma}

\begin{proof}[Proof of~\cref{lem:mop}]

    Consider the optimal resource distribution $\vx^*$ that obtains the maximum $\mlw(\vx^*)$.  
    The basic idea of this proof is to show that there exist exchange intervals $\{I_i^*\}_{i\in [n]}$ and price $\lambda^*$ that satisfies the equilibrium property and reaches $\vx^*$.

    Refer to the set of agents with $x_i^*\geq 0$  as the buyer set $\cB^*$ and the set of other agents as the seller set $\cS^*$. It can be observed that for any buyer $i\in \cB^*$ and seller $j\in \cS^*$, we have $v_i \geq v_j$; otherwise, shifting resources from $i$ to $j$ can increase the market liquid welfare.
    
    Without loss of generality, we can assume that any agent $i\in \cB^*$ buys at most $\frac{B_i}{v_i}$ units of resources in the optimal solution because, beyond this upper bound, the agent's contribution to the market liquid welfare will not increase further.

     For each buyer $i$ and seller $j$, let $I_i^* = [0,x_i^*]$ and $I_j^*=[x_j^*,0]$.
     Define $\lambda^*$ as the minimum valuation of buyers. Since $x_i^*\leq \frac{B_i}{v_i} \leq \frac{B_i}{\lambda^*}$ for each buyer $i\in \cB^*$, we have
     \[ \sum_{i\in \cB^*} \min\left\{x_i^*, \frac{B_i}{\lambda^*}\right\} = \sum_{i\in \cB^*} x_i^* = \sum_{i\in \cS^*} -x_i^* = \sum_{i\in \cS^*} \min\{-x_i^*, \Gamma_i\},  \]
     where the last two equalities used the fact that $\vx^*$ is self-consistent and feasible.

    Although the equation above is similar to the condition in~\cref{lem:equilibrium_property}, they are not exactly the same due to the definitions of $(\cB,\cS)$. In the seller set $\cS^*$ of the optimal solution, it is possible to have agents with $v_i = \lambda^*$. Therefore, we need to discuss this specific case separately. Let $\cP$ represent the set of agents with $v_i = \lambda^*$. Define the difference $\Delta:=  \sum_{i\in \cS^* \setminus \cP} \min\{-x_i^*, \Gamma_i\}- \sum_{i\in \cB^*\setminus \cP} \min\left\{x_i^*, \frac{B_i}{\lambda^*}\right\}$. We do the following refinement:

    \begin{itemize}[left=3em]
        \item If $\Delta \geq 0 $, reset the interval of each agent $i\in \cP$ such that $I_i^*=[0, r_i^*]$ with $r_i^* \leq B_i/\lambda^*$ and $\sum_{i\in \cP} r_i^* = \Delta$. 
        \vspace{0.5em}
        \item If $\Delta < 0 $, we first increase the price $\lambda^*$ slightly so that the minimum valuation in $\cB^* \setminus \cP$ remains (strictly) greater than it. Then reset the interval of each agent $i\in \cP$ such that $I_i^*=[-l_i^*, 0]$ with $l_i^* \leq \Gamma_i$ and $\sum_{i\in \cP} l_i^* = -\Delta$.
    \end{itemize}

     It can be observed the new constructed $\{I_i^*\}_{i\in [n]}$ and $\lambda^*$ satisfy the condition stated in~\cref{lem:equilibrium_property} without changing the objective. Therefore, using~\cref{lem:equilibrium_property}, we can complete this proof.
\end{proof}

From the analysis in the proof above, we further characterize the optimal distribution $\vx^*$. The proof is deferred to~\cref{sec:lem:optallo}.

\begin{lemma}[Optimal Resource Distribution]\label{lem:optallo}
    Given an agent set $\vA =  \{(v_i,B_i,\Gamma_i)\}_{i\in [n]}$ with $v_1\geq v_2 \geq ... \geq v_n$, the following resource distribution achieves the maximum market liquid welfare:
        $$ x_i^*=\begin{cases}
        B_i/v_i & i\leq k^*\\
        \sum_{i=k^*+2}^n \Gamma_i - \sum_{i=1}^{k^*} B_i/v_i & i=k^*+1\\
        -\Gamma_i & i\geq k^*+2
      \end{cases},$$
      where $k^*=\max\{l\in[n]\mid \sum_{i=1}^l B_i/v_i \le \sum_{i=l+1}^n \Gamma_i\}$.
\end{lemma}

\section{Approximating MOP Truthfully on Large Markets}\label{sec:mp}

\cref{lem:mop} proves the existence of a uniform-price mechanism that achieves the optimal market liquid welfare. However, it is evident that the constructed mechanism is not truthful, as buyers can increase their resource purchase limit by misreporting a lower valuation. Moreover, even if we disregard the issue raised by exchange intervals, determining the uniform price truthfully based on the reported valuations remains a challenging task. Fortunately, we have discovered that these issues can be successfully resolved under the mild large market assumption, with only a constant factor loss in approximation. Formally, we aim to show the following.

\begin{definition}[Large Market Assumption]\label{def:large_market}
Given a set of agent $\vA=\{v_i,B_i,\Gamma_i\}_{i\in [n]}$, define $\opt$ as the optimal market liquid welfare and $\vx^*$ as the resource distribution constructed by~\cref{lem:optallo}. Further, define the buyer set $\cB^* :=\{i \in [n] \mid x_i^* \ge 0\} $ and the seller set $\cS^* := \{i \in [n] \mid x_i^* < 0\}$.The market is considered 
\emph{large} with respect to a large constant $1/\theta$ if it satisfies the following conditions:
\begin{enumerate}[left=3em]
    \item[(1)] Limited contribution of each agent: $\forall i\in [n]$, $v_i \cdot \Gamma_i + \min\{v_i\cdot x_i^*,B_i\} \leq \theta \cdot \opt $.
    
    \item[(2)] No monopolistic buyer: $\forall i\in \cB^*$, $x_i^* \leq \theta \cdot(\sum_{j\in \cB^*}x_j^*) $.
    
    \item[(3)] No monopolistic seller: $\forall i\in \cS^*$, $ |x_i^*| \leq \theta \cdot (\sum_{j\in \cS^*}|x_j^*|) $.
\end{enumerate}

\end{definition}

\begin{theorem}\label{thm:large_market}
   Under the large market assumption (parameterized by $\theta$), there exists a randomized uniform price mechanism that is universally truthful and profitable, and admits constant approximation with high probability. Further, when parameter $\theta \rightarrow 0 $, the approximation ratio approaches $1/2$.
\end{theorem}

\subsection{Market $1/2$-Approximate Price}

We address the aforementioned issues of truthfulness one by one. This subsection focuses on addressing the untruthfulness that may arise from exchange intervals. We show that even without any constraint on the resource units bought or sold by each agent, a 2-approximation uniform price mechanism exists.

\begin{lemma}[Market Approximate Price]\label{lem:map}
    Given any agent set $\vA$, let $\opt$ be the optimal market liquid welfare.   
    There always exists a price $ \blam $ such that the obtained market liquid welfare is at least $\opt/2$ when each agent subjects to constraint $ ([-\infty,\infty], \blam) $.
\end{lemma}

\begin{proof}
    Let $\vgam = \sum_{i\in [n]} \Gamma_i$ be the sum units of resources in the market. We set the price $\blam = \frac{\opt}{2\cdot \vgam}$. Recall that the buyers $\cB$ are the agents with $v_i\geq \blam$ and the sellers $\cS$ are the agents with $v_i<\blam$.
    The analysis distinguishes two cases based on the equilibrium state $(\vx,\vp)$ that the market eventually reaches:
    \begin{enumerate}[left=3em]
        \item[(1)] No seller has remaining resources to sell.
        \vspace{0.5em}
        \item[(2)] No buyer is willing to buy.
    \end{enumerate}

    The proof of the first case is straightforward. As $x_j = -\Gamma_j$  for any seller $j\in \cS$ and any buyer $i\in \cB$ has $x_i = \frac{p_i}{\blam} \leq \frac{B_i}{\blam}$, we have
    \begin{align*}
        \mlw(\vx) = \sum_{i\in \cB} v_i \cdot \Gamma_i + \min\{v_i\cdot x_i, B_i\} 
        \geq \sum_{i\in \cB} \blam \cdot \Gamma_i + \min\{\blam \cdot x_i, B_i\} 
         = \vgam \cdot \blam \geq \frac{\opt}{2}.
    \end{align*}   

    For the second case, each buyer $i\in \cB$ exhausts the budget, i.e., $\forall i\in \cB$, $x_i = B_i/ \blam$, and achieves the maximum possible contribution to the objective, i.e., $B_i + v_i\cdot \Gamma_i$. As any seller has a value $v_i<\blam$, the maximum possible of their contributions is at most $\vgam \cdot \blam$. Then we get an upper bound of $\opt$: $\opt \leq \sum_{i\in \cB} v_i\cdot \Gamma_i + B_i + \vgam \cdot \blam$. Therefore, the proof can be completed:
    \begin{align*}
        \mlw(\vx) &\geq \sum_{i\in \cB} v_i\cdot \Gamma_i + B_i \geq \opt - \vgam \cdot \blam = \frac{\opt}{2}.
    \end{align*} 
\end{proof}

We further give a lower bound instance to demonstrate that $1/2$-approximation is the best possible for such mechanisms. Due to space limitations, the proof is deferred to~\cref{sec:lem:map_lower_bound}.

\begin{lemma}\label{lem:map_lower_bound}
    There exists an instance such that for any price $\lambda$ and any parameter $\epsilon>0$, the market liquid welfare is at most $(\frac{1}{2}+\epsilon)\opt$ when each agent subjects to constraint $([-\infty,\infty] , \lambda)$.
\end{lemma}

\subsection{Truthful Uniform Price Mechanism}

The last subsection eliminates the interval untruthfulness issue. Now we employ random sampling to set a uniform price truthfully. 

\begin{algorithm}[htb]
  \caption{Truthful \MECHL}\label{alg:uniform_large}
  \begin{algorithmic}[1]
  \Require Agent set $\vA=\{(v_i, B_i, \Gamma_i)\}_{i\in [n]}$ and parameter $\beta\in (0,1/2)$.
  \Ensure Exchange constraints $(\vI,\vlam)$  
  \State Independently sample each agent with a probability of $\beta$. Denote the agents that are sampled as $L$ and the remaining agents as $R$.
  \State Compute the optimal market liquid welfare $\opt(L)$ of the agent set $L$.
  \State Set a uniform price $\lambda \gets \frac{1-\beta}{\beta } \cdot \frac{\opt(L)}{2 \cdot \vgam_{R}}$, where $\vgam_R = \sum_{i\in R}\Gamma_i$ is the total initial resource units of agent set $R$.
  \State Set $I_i\gets [0,0]$ $\forall i\in L$ and $I_i \gets [-\infty,\infty]$ $\forall i\in R$. 
  \State \Return $\{(I_i, \lambda)\}_{i\in [n]}$
  \end{algorithmic}
\end{algorithm}

We describe the mechanism in~\cref{alg:uniform_large}. The main intuition is that when the market is large enough, we can estimate the price stated in~\cref{lem:map} accurately by sampling a small fraction of agents. We show the following two lemmas:

\begin{lemma}\label{lem:uniform_large_truthful}
    \cref{alg:uniform_large} is universally  truthful and profitable.
\end{lemma}

\begin{lemma}\label{lem:uniform_large_ratio}
    Under the large market assumption that $\theta = \beta/ c$ (a large constant $c\gg 1$), \cref{alg:uniform_large} obtains $\frac{1-\beta}{2}\cdot (1-O(\delta))$ approximation with probability at least $1-6\exp(-c\delta^2/3)$, where $\delta$ is a small constant compared to $c$.
\end{lemma}

The proofs are deferred to~\cref{sec:lem:uniform_large_truthful} and~\cref{sec:lem:uniform_large_ratio}. 
The truthfulness proof is straightforward due to the independent sampling, while for the approximation proof, we employ the optimal structure stated in~\cref{lem:optallo} and Chernoff bounds to show that with high probability, the sampled $\opt(L)$ is a good estimator of $\beta\cdot \opt$.  
Combining~\cref{lem:uniform_large_truthful} and~\cref{lem:uniform_large_ratio} proves~\cref{thm:large_market} immediately. Moreover, by \cref{lem:map_lower_bound}, we see that this result is almost tight.

\subsubsection*{Beyond the Single-Parameter Environment.} \cref{alg:uniform_large} can further be extended to a multi-parameter setting where all $B_i$ and $\Gamma_i$ are private. By adding an extra assumption that every $\Gamma_i$ is small compared to the sum, the same approximation can be obtained by a simple algorithm variant. A more detailed discussion is provided in~\cref{sec:app-mp}.

\section{\MECHS with Bounded Unprofitability}\label{sec:sp}

\Mechl achieves a good approximation ratio in large markets. In this section, we aim to design an alternative exchange market mechanism that performs well in not-so-large market environments, particularly when \(\theta\) is relatively large. Additionally, to ensure truthfulness, we allow this mechanism to have bounded unprofitability.
Our main results are presented in \cref{thm:mechs}.

\begin{theorem}\label{thm:mechs}
There exists a truthful exchange market mechanism that can achieve an approximation ratio of $1/2$, and the required subsidy does not exceed the total utility obtained by all agents under this mechanism.
\end{theorem}

\subsection{Mechanism Design}

In this subsection, we propose a \mechs and will subsequently prove that it meets the requirements specified in \cref{thm:mechs}.
Although agents do not trade according to a uniform price in this mechanism, it still adopts the idea from market uniform price and uniform price auction \cite{dobzinski2014efficiency}.
The mechanism divides buyers and sellers based on a uniform price and sets the exchange interval \(I_i\) for each agent accordingly. Nevertheless, the exchange price \(\lambda_i\) for each agent is not directly the uniform price itself but is instead a differentiated price for each individual.

Before describing \mechs in detail, we first define the notation $k$ and $q$ and the functions $x_i(v_i, \vecv_{-i})$ and $p_i(v_i, \vecv_{-i})$ that are used in the mechanisms. 
For simplicity, when emphasizing agent $i$, we will abbreviate $x_i(v_i, \vecv_{-i})$ and $p_i(v_i, \vecv_{-i})$ as $x_i(v_i)$ and $p_i(v_i)$, respectively.
Assume that the agents are ordered in descending order of their valuations, i.e., $v_1 \ge v_2 \ge \cdots \ge v_n$. Let $k = \max\left\{l \in [n] \mid \sum_{i=1}^l B_i \le v_l \cdot \sum_{i=l+1}^n \Gamma_i\right\}$, referred to as the \textit{partition point}, which is used to divide the agents into buyers and sellers.
Define 
$$
q = \begin{cases} 
    \frac{\sum_{i=1}^k B_i}{\sum_{i=k+1}^n \Gamma_i} & \text{if } \sum_{i=1}^k B_i > v_{k+1} \cdot \sum_{i=k+1}^n \Gamma_i \\
    v_{k+1} & \text{otherwise}
\end{cases}
$$
to mimic the role of a market uniform price.
The function $x_i(v_i, \vec{v}_{-i})$ is used to set the exchange interval, defined as follows:
\begin{equation}\label{eq:allo}
    x_i(v_i, \vecv_{-i}) = \begin{cases}
        \frac{B_i}{q} & \text{for } i = 1, 2, \ldots, k \\
        \sum_{i=k+2}^n \Gamma_i - \sum_{i=1}^k \frac{B_i}{q} & \text{for } i = k+1 \\
        -\Gamma_i & \text{for } i = k+2, \ldots, n
    \end{cases}.
\end{equation}
And the function $p_i$ is used to set the exchange price, defined as follows:
\begin{equation}\label{payment}
  p_i(v_i, \vecv_{-i}) = v_i x_i(v_i, \vecv_{-i}) - \int_{\hat{\vecv}_{-i}}^{v_i} x_i(z, {\bf v}_{-i}) \dif z,
\end{equation}
where $\hat{\vecv}_{-i}$ is the solution satisfying $x_i(\hat{\vecv}_{-i} - \delta, {\bf v}_{-i}) \le 0$ and $x_i(\hat{\vecv}_{-i} + \delta, {\bf v}_{-i}) \ge 0$, for any $\delta > 0$.
\footnote{According to \cref{lem:monotonicity}, and given that for any fixed ${\bf v}_{-i}$, $x_i(0, {\bf v}_{-i})=-\Gamma_i, x_i(+\infty, {\bf v}_{-i})=\frac{B_i}{q}$, $\hat{\vecv}_{-i}$ must exist. Furthermore, if there are multiple distinct values of $\hat{\vecv}_{-i}$, the corresponding $\int_{\hat{\vecv}_{-i}}^{v_i} x_i(z, {\bf v}_{-i}) \dif z$ are all identical. Therefore, $p_i(v_i, \vecv_{-i})$ is well-defined.}

Now, we can formally describe \mechs, as shown in \cref{alg:single}.

\begin{algorithm}[H]
  \caption{\MECHS}\label{alg:single}
  \begin{algorithmic}[1]
    \Require Agent set $\vA=\{(v_i, B_i, \Gamma_i)\}_{i\in [n]}$ \Comment{Assume $v_1 \ge v_2 \ge \cdots \ge v_n$}
  \Ensure Exchange constraints $(\vI,\vlam)$.

  \For{$i = 1$ to $n$}

    \State $I_i = \left[\min\left\{x_i(v_i),0\right\},\max\left\{x_i(v_i),0\right\}\right]$ \Comment{Using \cref{eq:allo} to calculate the function $x_i$.}

    \If{$x_i(v_i)\neq 0$}
        \State $\lambda_i=p_i(v_i)/x_i(v_i)$ \Comment{Using \cref{payment} to calculate the function $p_i$.}
    \Else
        \State $\lambda_i = v_i$
    \EndIf
    
  \EndFor
  \State \Return $\{(I_i, \lambda_i)\}_{i\in [n]}$

  \end{algorithmic}
\end{algorithm}

According to \cref{lem:monotonicity}, it can be easily proven that $x_i(v_i) p_i(v_i) \ge 0$, and $x_i(v_i) = 0$ if and only if $p_i(v_i) = 0$. Therefore, when $x_i(v_i) \neq 0$, we have $\lambda_i = \frac{p_i(v_i)}{x_i(v_i)} > 0$.
Consequently, for all $i$, $\lambda_i \ge 0$, which ensures that the exchange price $\lambda_i$ is well-defined.

\subsection{Market Outcomes}\label{sec:market_outcome_mechs}

This subsection analyzes the truthfulness of \mechs and the reachable market state $(\vx,\vp)$ that agents can achieve through free trade under this mechanism.
First, we present the monotonicity of $x_i(v_i, \vecv_{-i})$ and some properties of \mechs, as shown in \cref{lem:monotonicity,lem:several_properties}. The detailed proofs are provided in~\cref{sec:lem:monotonicity,sec:lem:several_properties}.

\begin{lemma}[Monotonicity]\label{lem:monotonicity}
The function $x_i(v_i, \vecv_{-i})$ is non-decreasing with respect to $v_i$. Specifically, for any agent $i$, an increase in $v_i$ to $v_i'$, while keeping $\vecv_{-i}$ unchanged, results in an allocation where $x_i(v_i', \vec{v}_{-i}) \geq x_i(v_i, \vecv_{-i})$.
\end{lemma}

\begin{lemma}\label{lem:several_properties}
For all $i \in [n]$, for all $v_i \in \R_+$, and for all $\vecv_{-i} \in \R_+^{n-1}$, the following properties hold:
\begin{itemize}[left=3em]
    \item[(1)] $(v_i - \lambda_i) x_i(v_i) \ge 0$. Therefore, when $v_i \ge \lambda_i$, $I_i=[0,x_i(v_i)] \subset \R_+$, i.e., agent $i$ can only act as a buyer; and when $v_i < \lambda_i$, $I_i=[x_i(v_i),0] \subset \R_-$, i.e., agent $i$ can only act as a seller;
    \item[(2)] $\frac{B_i}{\lambda_i} \ge x_i(v_i)$ and $x_i(v_i) \ge -\Gamma_i$.
\end{itemize}
\end{lemma}

Based on the above Lemmas, we can begin to analyze the utility of any agent $i$ when reporting their valuation truthfully and when misreporting. Fix the valuations of the other agents as $\vecv_{-i}$.

When the agent truthfully reports their valuation $v_i$, let the exchange interval and exchange price output by the mechanism be denoted as $\vI=(I_1,I_2,\cdots,I_n)$ and $\vlam=(\lambda_1,\lambda_2,\cdots,\lambda_n)$, respectively. For any $j \in [n]$, define $x_j = x_j(v_j, \vecv_{-j})$ and $p_j = p_j(v_j, \vecv_{-j})$.

Let $\cB = \{j \mid v_j \ge \lambda_j\}$ and $\cS = \{j \mid v_j < \lambda_j\}$. According to Lemma~\ref{lem:several_properties}, we have 
$$
\sum_{j \in \cB} \min\{x_j, B_j/\lambda_j\} + \sum_{j \in \cS} \max\{x_j, -\Gamma_j\} = \sum_{j \in \cB} x_j + \sum_{j \in \cS} x_j = \sum_{j \in [n]} x_j = 0.
$$
From \cref{lem:equilibrium_property}, we can deduce that the reachable resource distribution $\vx$ must be $\vx = \{x_1, x_2, \cdots, x_n\}$ and the payment $\vp$ must be $\vp = \{p_1,p_2,\cdots,p_n\}$, where $p_i=\lambda_i x_i$ for all $i\in [n]$.
Therefore, the utility of agent $i$ is $u_i=u_i(\vx, \vp) = v_i x_i - p_i = v_i x_i - \lambda_i x_i$.

Next, we analyze the case when agent $i$ misreports their valuation as $v_i'$. Let $\vecv' = (v_i', \vecv_{-i})$. The exchange interval and exchange price output by the mechanism with input $\vecv'$ are denoted as $\vI' = (I_1', I_2', \cdots, I_n')$ and $\vlam' = (\lambda_1', \lambda_2', \cdots, \lambda_n')$, respectively. For any $j \in [n]$, define $x_j' = x_j(\vecv')$ and $p_j' = p_j(\vecv')$.

According to part (1) of \cref{lem:several_properties}, an agent $j\neq i$ who reports truthfully will desire to trade a quantity $x_j'$ of the resource to maximize their utility. Given that $\sum_{j \in [n]} x_j' = 0$ still holds, agent $i$ can choose to trade any quantity of the resource within the interval $I_i = \left[\min\left\{x_i', 0\right\}, \max\left\{x_i', 0\right\}\right]$.

When $v_i$ and $v_i' \ge \lambda_i'$, agent $i$ wants to buy as much resource as possible. According to \cref{lem:several_properties}, part (1), in this case, $I_i' \subset \R_+$, and thus $u_i' = (v_i - \lambda_i') \cdot x_i' = v_i x_i' - p_i'$. When $v_i \ge \lambda_i'$ and $v_i' < \lambda_i'$, agent $i$ still wants to buy as much resource as possible. However, according to \cref{lem:several_properties}, part (1), in this case, $I_i' \subset \R_-$, meaning the agent can only sell resources. Consequently, agent $i$ will choose not to sell, resulting in $u_i' = 0$.
Similarly, when $v_i$ and $v_i' < \lambda_i'$, $u_i' = v_i x_i' - p_i'$. When $v_i < \lambda_i'$ and $v_i' \ge \lambda_i'$, $u_i' = 0$.

With these analyses, we can now begin to prove \cref{thm:truthfulness}.

\begin{lemma}
    [Truthfulness]~\label{thm:truthfulness}
  \Mechs is truthful.
\end{lemma}

\begin{proof}
Based on the above analysis, to prove truthfulness, it suffices to show that $u_i \ge u_i'$. Although $u_i'$ has different expressions under different conditions, we can prove a stronger conclusion, namely that $u_i$ is greater than or equal to $u_i'$ in any form. This means proving that $v_i x_i - p_i \ge v_i x_i' - p_i'$ as well as $v_i x_i - p_i \ge 0$. The inequality $v_i x_i - p_i \ge 0$ can be directly obtained from part (1) of Lemma~\ref{lem:several_properties}. Next, we prove that $v_i x_i - p_i \ge v_i x_i' - p_i'$.

  According to \cref{lem:monotonicity}, $x_i(v_i,\vecv_{-i})$ is monotonic. 
  Furthermore, the payment can be written as 
\begin{equation*}
   p_i = v_ix_i - \int_{\hat{\vecv}_i}^{v_i} x(z, {\bf v}_{-i}) \dif z = v_ix_i - \int_0^{v_i} x(z, {\bf v}_{-i}) \dif z + h({\bf v}_{-i}),  
\end{equation*}
  where $ h({\bf v}_{-i})=\int_0^{\hat{\vecv}_{-i}} x(z, {\bf v}_{-i}) \dif z$ is a term independent of $v_i$. Therefore, according to Myerson's Lemma (\cref{lem:Myerson}), $v_i x_i(v_i, \vecv_{-i}) - p_i(v_i, \vecv_{-i}) \ge v_i x_i(v_i', \vecv_{-i}) - p_i(v_i', \vecv_{-i})$, i.e., $v_i x_i - p_i \ge v_i x_i' - p_i'$.
\end{proof}

Within the framework of \mechs, it is conceivable that \(\sum_{i \in [n]} p_i \le 0\), indicating that the market manager, such as the government, may need to provide a subsidy to implement this mechanism and achieve a higher welfare. \cref{lem:boundC} shows that the amount of the subsidy is bounded. The proof can be found in~\cref{sec:lem:boundC}.

\begin{lemma}[Bounded Compensation]\label{lem:boundC}
The compensation amount does not exceed the total utility obtained by all agents, specifically, \(\sum_{i \in [n]} -p_i \leq \sum_{i \in [n]} u_i\), where $u_i=v_i x_i-p_i$.
\end{lemma}

\subsection{Approximation Ratio}

Finally, we analyze the approximation ratio. 

\begin{lemma}\label{thm:2approx}
  The \mechs is $1/2$-approximation.
\end{lemma}

\begin{proof}
    Suppose the indices of agents are ordered in descending order of valuations, i.e., $v_1 \geq v_2 \geq \cdots \geq v_n$. Let $\vx = (x_1, x_2, \cdots, x_n)$ represent the distribution obtained from \mechs. According to the analysis in \cref{sec:market_outcome_mechs}, we have $x_i = x_i(v_i)$.
For an optimal resource distribution $\vx^*=(x_1^*,x_2^*,\cdots,x_n^*)$, denote the optimal welfare by $\opt$, i.e. $\opt=\mlw(\vx^*)$.
Therefore, we have:
  \begin{eqnarray}
\opt &=& \sum_{i=1}^n v_i \Gamma_i + \sum_{i=1}^n \min\{v_i x_i^*, B_i\}
      \le \sum_{i=1}^k (B_i + v_i \Gamma_i) + \sum_{i=k+1}^n v_i(\Gamma_i+x_i^*)\nonumber\\
      &\le& \sum_{i=1}^k (B_i + v_i \Gamma_i) + v_{k+1} \cdot \sum_{i=k+1}^n \Gamma_i,\label{eqn:OPT's upper bound}
  \end{eqnarray}
where $k$ is defined in \mechs, and the first inequality is due to then property of the $\min$ function, the second inequality follows from \cref{lem:optallo}. From the mechanism, 
  \begin{equation}\label{eqn:GLW's Lower bound1}
  \begin{aligned}
      \mlw(\vx)&= \sum_{i=1}^n \min\{v_i x_i,B_i\} + \sum_{i=1}^n v_i \Gamma_i \\
      & = \sum_{i=1}^k (B_i+v_i \Gamma_i) + v_{k+1} ( x_{k+1} + \Gamma_{k+1} )\geq \sum_{i=1}^k (B_i+v_i \Gamma_i),
  \end{aligned}
  \end{equation}
  where the last inequality is from part (2) of \cref{lem:several_properties}.

If $\sum_{i=1}^k B_i>v_{k+1}\cdot\sum_{i=k+1}^n \Gamma_i$, then $x_{k+1} = -\Gamma_{k+1}$. Thus,
  \begin{eqnarray}
      \mlw(\vx) = \sum_{i=1}^k (B_i+v_i \Gamma_i) > v_{k+1}\sum_{i=k+1}^n \Gamma_i + \sum_{i=1}^k v_i \Gamma_i \ge v_{k+1}\sum_{i=k+1}^n \Gamma_i.\label{eqn:GLW's Lower bound2}
  \end{eqnarray}

If $\sum_{i=1}^k B_i\le v_{k+1}\sum_{i=k+1}^n \Gamma_i$, then $x_{k+1} = \sum_{i=k+2}^n \Gamma_i - \sum_{i=1}^k B_i/v_{k+1}$. 
  \begin{eqnarray}
  \mlw(\vx) &=& \sum_{i=1}^k B_i + v_{k+1}  \left( \sum_{i=k+2}^n \Gamma_i - \sum_{i=1}^k B_i/v_{k+1} \right) + \sum_{i=1}^{k+1} v_i \Gamma_i \nonumber\\
      &=&
    v_{k+1}\sum_{i=k+1}^n \Gamma_i + \sum_{i=1}^k v_i \Gamma_i 
      \ge v_{k+1} \sum_{i=k+1}^n \Gamma_i.\label{eqn:GLW's Lower bound3}  
  \end{eqnarray}
Therefore, by (\ref{eqn:OPT's upper bound}), (\ref{eqn:GLW's Lower bound1}), (\ref{eqn:GLW's Lower bound2}) and (\ref{eqn:GLW's Lower bound3}), we have
\begin{eqnarray*}
    \opt\leq \sum_{i=1}^k(B_i+v_i\Gamma_i)+v_{k+1}\sum_{i=k+1}^n\Gamma_i\leq 2\mlw(\vx).
\end{eqnarray*}
Consequently, this result holds.
\end{proof}

Combining \cref{thm:truthfulness}, \cref{lem:boundC} and \cref{thm:2approx} proves \cref{thm:mechs} directly.

\section{Conclusion}\label{sec:con}
This paper introduces a novel and practical exchange market model. Within this model, we prove several properties related to market equilibrium and optimal structure, and present two mechanisms with theoretical guarantees.
Notice that our mechanism for general markets is unprofitable. An open question we pose is whether a truthful, profitable, and constant-approximate mechanism exists in general markets. Intuitively, we lean towards answering ``no''. Consider a scenario with one buyer and one seller --- it appears challenging to set a uniform price that simultaneously satisfies truthfulness and constant approximation. However, rigorously proving this claim requires new analytical techniques.

\newpage
\bibliographystyle{splncs04}
\bibliography{mybibliography}

\appendix
\newpage

\section{Omitted Proofs in~\cref{sec:mop}}\label{sec:mop_proof}

\subsection{Proof of~\cref{lem:equilibrium_property}}\label{sec:leem:eq_pro}

The lemma can be proved directly by the reachability definition (\cref{def:reach}). As $\sum_{i\in \cB} \min\left\{r_i, \frac{B_i}{\lambda_i}\right\} = \sum_{j\in \cS} \min\left\{l_j, \Gamma_j\right\}$, the market's supply and demand are equal. Therefore, when the market reaches equilibrium, all buyers reach their upper bounds for purchases, and all sellers reach their upper bounds for sales, imply a unique reachable state and the welfare claimed in the lemma.

\subsection{Proof of~\cref{lem:optallo}}\label{sec:lem:optallo}

Denote by $\vx^*$ the resource distribution defined in the lemma.
Considering any other feasible $\vx' =\{x_i'\}_{i\in [n]}$, we show that $\mlw(\vx')\leq \mlw(\vx^*)$.

Without loss of generality, we can assume that for any agent $i$, $x_i'$ is at most $B_i/v_i$ because, beyond this upper bound, the agent's contribution to the market liquid welfare will not increase further. Then for each agent $i\leq [k^*]$, we have $x_i'\leq \frac{B_i}{v_i}=x_i^*$.
According to the construction of $\vx^*$, we also have for any agent $i\geq k^*+2$, $ x_i' \geq -\Gamma_i = x_i^*$.

Combining the above two inequalities and the fact that $x_{k^*+1}' =-\sum_{i=1}^{k^*} x_i' 
-\sum_{i=k^*+2}^{n} x_i' $, then 
\begin{align*}
      \mlw(\vx')-\sum_{i\in [n]} v_i\cdot \Gamma_i 
      = & \sum_{i=1}^{k^*} \min\{v_i \cdot x_i',B_i\} + \min\{v_{k^*+1} \cdot x_{k^*+1}',B_{k^*+1}\} + \sum_{i=k^*+2}^{n} \min\{v_i\cdot x_i',B_i\}\\
    \leq & \sum_{i=1}^{k^*} v_i\cdot x_i' - v_{k^*+1} \cdot \left(\sum_{i=1}^{k^*} x_i' + \sum_{i=k^*+2}^{n} x_i'\right)  +
    \sum_{i=k^*+2}^{n} v_i \cdot x_i'\\
    =&\sum_{i=1}^{k^*} (v_i-v_{k^*+1})x_i' -
    \sum_{i=k^*+2}^{n} (v_{k^*+1}-v_i) x_i'\\
    \leq & \sum_{i=1}^{k^*} (v_i-v_{k^*+1})x_i^* -
    \sum_{i=k^*+2}^{n} (v_{k^*+1}-v_i) x_i^*\\
    =& \mlw(\vx^*) - \sum_{i\in [n]} v_i\cdot \Gamma_i.
\end{align*}

\section{Omitted Proofs in~\cref{sec:mp}}\label{sec:mp_proof}

\subsection{Proof of~\cref{lem:map_lower_bound}}\label{sec:lem:map_lower_bound}

The claimed instance consists of three agents as follows:
\begin{itemize}[left=3em]
    \item $v_1= \frac{1/2 +\epsilon}{2\epsilon}$, $B_1 = 1$, $\Gamma_1 = 0$
    \item $v_2= 1$, $B_2 = 1$, $\Gamma_2 = 0$
    \item $v_3 = 0$, $B_3=0$, $\Gamma_3 = 1$.
\end{itemize}

Clearly, the optimal market liquid welfare is $1/(1/2+\epsilon)$. However, for any exchange constraint $([-\infty,\infty], \lambda)$, the market liquid welfare of the worst reachable state is $1$. If $\lambda > 1$, only the first agent can buy; while if $\lambda\leq 1$, both agent 1 and agent 2 can buy, but in the worst case, the second agent buys all resources and the market liquid welfare is still 1. Thus, the approximation ratio is at most $1/2+\epsilon$.

\subsection{Proof of~\cref{lem:uniform_large_truthful}}\label{sec:lem:uniform_large_truthful}

The proof is straightforward. In any scenario, agents in $L$ are unable to engage in trades, while the reported valuations of agents in $R$ do not affect the uniform price $\lambda$. As a result, no one has an incentive to lie. Then due to the uniform price property of the algorithm, the market is always profitable.

\subsection{Proof of~\cref{lem:uniform_large_ratio}}\label{sec:lem:uniform_large_ratio}

    Consider the optimal resource distribution $\vx^*$ constructed by~\cref{lem:optallo}.
    Let $\cB^*$ and $\cS^*$ be the corresponding buyers and sellers.
    Define $\mlw(\vx^*_L)$ (resp. $\mlw(\vx^*_R)$) as the total contribution of agents in $L$ (resp. $R$) to $\opt$. The basic idea of this proof is to show that $\opt(L)/\beta$, $\mlw(\vx^*_L)/\beta$, and $\opt$ are close to each other with high probability. As a result, the computed $\lambda$ by the algorithm approaches the approximate market price stated in~\cref{lem:map}, and a constant approximation can be achieved. Since agents in $L$ are sampled independently with a probability of $\beta$, we have 
   $  \E[\mlw(\vx^*_L)] = \beta \cdot \opt. $
   
    Using the Chernoff bound and the large market assumption, we have for any $\delta \in (0,1)$, 
    \[\Pr\left[ | \mlw(\vx^*_L) - \beta \cdot \opt | \geq \delta \cdot \beta \cdot \opt \right] \leq 2\exp\left(-\frac{c \delta^2}{3}\right).\]
    Similarly, we have
    \[\Pr\left[ | \vx_{L\cap \cB^*}^* - \beta \cdot \vx_{\cB^*}^* | \geq \delta \cdot \beta \cdot \vx_{\cB^*}^* \right] \leq 2\exp\left(-\frac{c\delta^2}{3}\right),\]
    and
    \[\Pr\left[ | \vx_{L\cap \cS^*}^* - \beta \cdot \vx_{\cS^*}^* | \geq \delta \cdot \beta \cdot |\vx_{\cS^*}^*| \right] \leq 2\exp\left(-\frac{c\delta^2}{3}\right).\]
    where for simplicity, we slightly abuse the notation by letting $\vx_{Q}^*$ for an agent subset $Q$ represent both the set $\{x_i^*\}_{i\in Q}$ and the sum of the elements in $Q$.

By the union bound, we have with probability at least $1-6\exp(-c\delta^2/3)$, the following three properties hold simultaneously:
\begin{enumerate}[left=3em]
    \item[(P1)] $\mlw(\vx^*_L) \in [(1-\delta)\beta\cdot \opt, (1+\delta)\beta \cdot \opt ] $
    \vspace{0.5em}
    \item[(P2)] $\vx_{L\cap \cB^*}^* \in [ (1-\delta)\beta\cdot \vx^*_{\cB^*}, (1+\delta)\beta\cdot \vx^*_{\cB^*} ]$
    \vspace{0.5em}
    \item[(P3)] $\vx_{L\cap \cS^*}^* \in [ (1-\delta)\beta\cdot |\vx^*_{\cS^*}|, (1+\delta)\beta\cdot |\vx^*_{\cS^*}| ]$
\end{enumerate}    

Now we are ready to establish the relationship between $\opt(L)$ and $\mlw(\vx_L^*)$. Note that $\vx_L^*$ could be not self-consistent, i.e., $\vx_L^*\neq 0$. Fortunately, due to the fact that $\vx_{\cB^*}^* = |\vx_{\cS^*}^*|$ and the last two properties above, we have
\[ \frac{1-\delta}{1+\delta} \leq \frac{\vx_{L\cap \cB^*}^*}{|\vx_{L\cap \cS^*}^*|} \leq \frac{1+\delta}{1-\delta}. \]
Therefore, by decreasing each term in $\vx_{L\cap \cB^*}$ at most a factor of $\frac{1+\delta}{1-\delta}$, $\vx_{L}^*$ will become self-consistent. Observing that the objective also decreases at most $\frac{1+\delta}{1-\delta}$ factor, we have
\begin{equation}\label{eq:large:L1}
  \opt(L) \geq \frac{1-\delta}{1+\delta} \cdot \mlw(\vx_L^*).  
\end{equation}
On the other hand, if multiplying each term in $\vx_{L\cap \cS^*}^*$ and $\{\Gamma_i\}_{i\in L\cap \cS^*}$ by at most a factor of $\frac{1+\delta}{1-\delta}$, we can obtain a new instance with a larger optimal objective and $\vx_{L}^*$ becomes not only a self-consistent, but also an optimal resource distribution on it according to~\cref{lem:optallo}. Observing that the objective increases at most $\frac{1+\delta}{1-\delta}$ during this process, we have
\begin{equation}\label{eq:large:L2}
     \opt(L) \leq \frac{1+\delta}{1-\delta} \cdot \mlw(\vx_L^*).
\end{equation}

Combing~\cref{eq:large:L1} and~\cref{eq:large:L2} and the property $(P1)$, the uniform price $\lambda$ in our algorithm can be bounded:
\begin{equation}\label{eq:large:lam}
    \lambda = \frac{1-\beta}{\beta } \cdot \frac{\opt(L)}{2 \cdot \vgam_{R}} \in \left[ (1-\beta) \cdot \frac{(1-\delta)^2}{1+\delta}\cdot \frac{\opt}{2 \vgam_R}, (1-\beta)\cdot \frac{(1+\delta)^2}{1-\delta} \cdot \frac{\opt}{2 \vgam_R} \right]
\end{equation}

Observe that $ \mlw(\vx_{R}^*) = \opt -  \mlw(\vx_{L}^*) $, $\vx_{R\cap \cB^*}^* =\vx_{\cB^*}^* -\vx_{L\cap \cB^*}^*  $ and  $\vx_{R\cap \cS^*}^* =\vx_{\cS^*}^* -\vx_{L\cap \cS^*}^* $. By symmetry, we have
\[ \frac{\opt(R)}{\opt} \in \left[ \frac{(1-\delta)^2}{1+\delta} \cdot (1-\beta),  \frac{(1+\delta)^2}{1-\delta} \cdot (1-\beta)\right]. \]

Combing the above with~\cref{eq:large:lam}, we have
\[ \lambda \in \left[ \frac{(1-\delta)^3}{(1+\delta)^3} \cdot \frac{\opt(R)}{2\Gamma_R}, \frac{(1+\delta)^3}{(1-\delta)^3} \cdot \frac{\opt(R)}{2\Gamma_R}  \right] \]

Let $(\vx,\vp)$ be an arbitrary reachable state of our algorithm. Similar to the two-case analysis in~\cref{lem:map}, we have for a small constant $\delta$,
\begin{align*}
    \mlw(\vx) &\geq \left(1- \frac{(1+\delta)^3}{2(1-\delta)^3} \right) \cdot \opt(R) 
     \geq (1-\beta) \cdot \frac{(1-\delta)^2}{1+\delta} \cdot \left(1- \frac{(1+\delta)^3}{2(1-\delta)^3} \right) \cdot \opt \\
    & = \frac{1-\beta}{2}\cdot (1-O(\delta)) \cdot \opt.
\end{align*}

\section{\MECHL for the Multiple-Parameter Setting}\label{sec:app-mp}

The goal of this section is to design a truthful market mechanism in a multi-parameter environment, where the valuation, budget, and initial resource amount are all private information of the agents.
However, in a multi-parameter environment, directly setting the uniform price as $\frac{1-\beta}{\beta}\cdot\frac{\opt(L)}{2\cdot \sum_{i\in R}\Gamma_i}$ will result in the mechanism not being truthful.
For instance, agents with high valuations could overreport their initial resources $\Gamma_i$ to secure a lower uniform price and purchase resources at a lower cost when allocated to $R$.
Therefore, a natural idea is to estimate the resources in $R$ using the resources in $L$. 
Specifically, we set the uniform price as $\frac{\opt(L)}{2\cdot \sum_{i \in L} \Gamma_i}$.
The detailed mechanism is outlined in \cref{alg:large_multiple}.

\begin{algorithm}[H]
  \caption{\MECHL for the Multiple-Parameter Setting}\label{alg:large_multiple}
  \begin{algorithmic}[1]
  \Require Agent set $\vA=\{(v_i, B_i, \Gamma_i)\}_{i\in [n]}$ and parameter $\beta\in (0,1/2)$.
  \Ensure Exchange constraints $(\vI,\vlam)$  
  \State Independently sample each agent with a probability of $\beta$. Denote the agents that are sampled as $L$ and the remaining agents as $R$.
  \State Compute the optimal market liquid welfare $\opt(L)$ of the agent set $L$.
  \State Set a uniform price $\lambda \gets \frac{\opt(L)}{2\cdot \vgam_L}$, where $\vgam_L = \sum_{i\in L}\Gamma_i$ is the total initial resource units of agent set $L$.
  \State Set $I_i\gets [0,0]$ $\forall i\in L$ and $I_i \gets [-\infty,\infty]$ $\forall i\in R$. 
  \State \Return $\{(I_i, \lambda)\}_{i\in [n]}$
  \end{algorithmic}
\end{algorithm}

Under this mechanism, to achieve a good approximation ratio, we need to slightly modify \cref{def:large_market} by adding another constraint on $\theta$:
\begin{itemize}[left=3em]
    \item[\textit{(4)}] \textit{No monopolistic resource holders: $\forall i \in [n], \Gamma_i \le \theta \cdot \left( \sum_{j \in [n]} \Gamma_j \right)$.}
\end{itemize}
The new constraint in our assumption is introduced to prevent significant deviations of resource quantities on sides $L$ and $R$ from their expected values during sampling. Such deviations could result in inaccurate estimations of the resource quantities on side $R$ based on those on side $L$.

Similar to \cref{alg:uniform_large}, \cref{alg:large_multiple} satisfies the following theorem. The proof is omitted due to its similarity to the proof of \cref{thm:large_market}. %

\begin{theorem}\label{thm:uniform_large_mp}
    \cref{alg:large_multiple} is universally truthful and profitable. 
    Under the modified large market assumption that $\theta = \beta/ c$ (a large constant $c \gg 1$), \cref{alg:large_multiple} obtains a $\frac{1-\beta}{2}\cdot (1-O(\delta))$ approximation with probability at least $1-8\exp(-c\delta^2/3)$, where $\delta$ is a small constant compared to $c$. Furthermore, when the parameter $\theta \rightarrow 0$, the approximation ratio approaches $1/2$.
\end{theorem}

\section{Omitted Proofs in~\cref{sec:sp}}\label{sec:proof}

\subsection{Proof of~\cref{lem:monotonicity}}\label{sec:lem:monotonicity}
For simplicity, in the proof of this lemma, we use $x_i$ to denote $x_i(v_i)$.

  Consider agent $i$. When $i \le k$, an increase in $v_i$ does not change $q$ and $x_i = B_i/q$ remains constant. For $i > k+1$, even as $v_i$ increases but remains less than $v_{k+1}$, $x_i = -\Gamma_i$ remains constant. Thus, it is only necessary to examine the scenario where $i = k+1$.

  Therefore, it is only necessary to consider the case where $i = k+1$.
  Given the definition of $k$, it is clear that $\max\{v_{k+2},\sum_{i=1}^kB_i/\sum_{i=k+1}^n \Gamma_i\}\le\min\{v_k,\sum_{i=1}^{k+1}B_i/\sum_{i=k+2}^n \Gamma_i\}$.
  When \\ $v_{k+1} \le \max\{v_{k+2},\sum_{i=1}^kB_i/\sum_{i=k+1}^n \Gamma_i\}$, $x_i = -\Gamma_i$. 
  When $v_{k+1} \in (\max\{v_{k+2},\sum_{i=1}^kB_i/\sum_{i=k+1}^n \Gamma_i\},\\ \min\{v_k,\sum_{i=1}^{k+1}B_i/\sum_{i=k+2}^n \Gamma_i\})$, $x_i = \sum_{i=k+2}^n \Gamma_i-\sum_{i=1}^k B_i/v_{k+1}$, which is monotonically increasing in $v_{k+1}$.
  When $v_{k+1} = \min\{v_k, \sum_{i=1}^{k+1}B_i / \sum_{i=k+2}^n \Gamma_i\}$, it is necessary to discuss the following two cases.

  \textbf{Case 1:} $v_k \ge \sum_{i=1}^{k+1}B_i / \sum_{i=k+2}^n \Gamma_i$. In this case, $v_{k+1} = \sum_{i=1}^{k+1}B_i / \sum_{i=k+2}^n \Gamma_i$. Consequently,\\ $\sum_{i=1}^{k+1}B_i \le v_{k+1} \cdot \sum_{i=k+2}^n \Gamma_i$. By the definition of $k$, $\sum_{i=1}^{k+2}B_i > v_{k+2} \cdot \sum_{i=k+3}^n \Gamma_i$, thus $x_{k+1} = B_{k+1}/q$. Similarly, according to the definition of $k$, $\sum_{i=1}^{k+1}B_i > v_{k+2} \cdot \sum_{i=k+2}^n \Gamma_i$, hence $p = \sum_{i=1}^{k+1}B_i / \sum_{i=k+2}^n \Gamma_i$. Therefore, $x_{k+1} = (B_{k+1} / \sum_{i=1}^{k+1}B_i)\cdot \sum_{i=k+2}^n \Gamma_i$. Given that $(B_{k+1} / \sum_{i=1}^{k+1}B_i) \cdot \sum_{i=k+2}^n \Gamma_i = \sum_{i=k+2}^n \Gamma_i - \sum_{i=1}^k B_i / (\sum_{i=1}^{k+1}B_i / \sum_{i=k+2}^n \Gamma_i)$, it follows that in this case, the change in $x_{k+1}$ is continuous.

  \textbf{Case 2:} $v_k < \sum_{i=1}^{k+1}B_i / \sum_{i=k+2}^n \Gamma_i$. In this scenario, $v_{k+1} = v_k$. Without loss of generality, when $v_{k+1} = v_k$, the tie-breaking rule gives preference to agent $k+1$ over $v_k$. 
  Let's delve into two distinct scenarios: 
  In the first scenario: $\sum_{i=1}^{k-1}B_i + B_{k+1} > v_k \cdot (\Gamma_k + \sum_{i=k+2}^n \Gamma_i)$.
  According to the definition of $k$, we have $\sum_{i=1}^{k-1}B_i \le v_{k-1} \cdot \sum_{i=k}^n \Gamma_i$, and $\sum_{i=1}^{k-1}B_i \le \sum_{i=1}^{k}B_i \le v_k \cdot \sum_{i=k+1}^n \Gamma_i \le v_k \cdot \sum_{i=k}^n \Gamma_i$. Therefore, $q = v_k$ and $x_{k+1} = \sum_{i=k+2}^n \Gamma_i + \Gamma_k - \sum_{i=1}^{k-1}B_i / v_k > \sum_{i=k+2}^n \Gamma_i - \sum_{i=1}^{k}B_i / v_k$. This implies that $x_{k+1}$ exhibits an increase at the point where $v_{k+1} = v_k$.
  In the second scenario: $\sum_{i=1}^{k-1}B_i + B_{k+1} \le v_k \cdot (\Gamma_k + \sum_{i=k+2}^n \Gamma_i)$.
  Given that $v_k < \sum_{i=1}^{k+1}B_i / \sum_{i=k+2}^n \Gamma_i$, it follows that $\sum_{i=1}^{k+1} B_i > v_k \cdot \sum_{i=k+2}^n \Gamma_i$. Therefore, we find $q = v_k$ and $x_{k+1} = B_{k+1} / v_k > \sum_{i=k+2}^n \Gamma_i - \sum_{i=1}^{k}B_i / v_k$. As a result, $x_{k+1}$ exhibits an increase at the point where $v_{k+1} = v_k$.

\subsection{Proof of~\cref{lem:several_properties}}\label{sec:lem:several_properties}
Firstly, we prove part (1), $(v_i - \lambda_i) x_i(v_i) \ge 0$, i.e., $v_i x_i(v_i) - p_i(v_i) \ge 0$.
According to \cref{payment}, we have $v_i x_i(v_i) - p_i(v_i) = \int_0^{v_i} x(z, \vecv_{-i}) \dif z - h(\vecv_{-i})$, where $h(\vecv_{-i}) = \int_0^{\hat{\vecv}_{-i}} x(z, \vecv_{-i}) \dif z$. By the definition of $\hat{\vecv}_{-i}$ and monotonicity of $x_i(v_i)$, we have $h(\vecv_{-i}) = \min_v \int_0^v x(z, \vecv_{-i}) \dif z \le \int_0^{v_i} x(z, \vecv_{-i}) \dif z$. Therefore, $v_i x_i(v_i) - p_i(v_i) \ge 0$ is proven.

Secondly, we prove part (2), $\frac{B_i}{\lambda_i} \ge x_i(v_i)$ and $x_i(v_i) \ge -\Gamma_i$. According to the definition of $x_i(v_i)$ (\cref{eq:allo}), it is evident that $x_i(v_i) \ge -\Gamma_i$. Next, we prove $\frac{B_i}{\lambda_i} \ge x_i(v_i)$, which is equivalent to proving $B_i \ge \lambda_i x_i(v_i) = p_i(v_i)$.

According to \cref{lem:monotonicity}, $x_i(v_i)$ is a non-decreasing function.
For agent $i$ with $x_i(v_i)<0$, by the definition of $\hat{\vecv}_{-i}$, we have $v_i\leq \hat{\vecv}_{-i}$ and $x_i(z,{\bf v}_{-i})\leq 0$ for any $z\in [v_i, \hat{\vecv}_{-i}]$, implying that $p_i=v_ix_i-\int_{\hat{\vecv}_{-i}}^{v_i}x_i(z,{\bf v}_{-i})\dif z\leq 0\leq B_i$. For case of $x_i(v_i)=0$, it is evident that $p_i=0\leq B_i$. For the case of $x_i(v_i)>0$, it implies $i\in [k+1]$. If $i\in [k]$, then $\hat{\vecv}_{-i}\le q\le v_i$ and $x_i(z)=\frac{B_i}{q}$ for $z\in [q,v_i]$. So $x_i(z,{\bf v}_{-i})\geq 0$ for $z\in [\hat{\vecv}_{-i},q]$, thus
\begin{eqnarray*}
  p_i(v_i)&=&v_ix_i(v_i)-\int_{\hat{\vecv}_{-i}}^{v_i}x_i(z,{\bf v}_{-i})\dif z
  =v_ix_i(v_i)-\int_{\hat{\vecv}_{-i}}^{q}x_i(z,{\bf v}_{-i})\dif z-\int_{q}^{v_i}x_i(z,{\bf v}_{-i})\dif z\\
  &=&v_i\frac{B_i}{q}-\int_{\hat{\vecv}_{-i}}^{q}x_i(z,{\bf v}_{-i})\dif z-\frac{B_i}{q}(v_i-q)
  =B_i-\int_{\hat{\vecv}_{-i}}^{q}x_i(z,{\bf v}_{-i})\dif z\leq B_i.
\end{eqnarray*}
If $i=k+1$, then we claim that $\sum_{i=1}^kB_i\leq v_{k+1}\cdot \sum_{k+1}^n\Gamma_i$. Otherwise, $q=\frac{\sum_{i=1}^kB_i}{\sum_{i=k+1}^n\Gamma_i}$ and $x_{k+1}(v_{k+1})=\sum_{i=k+2}^n\Gamma_i-\sum_{i=1}^k\frac{B_i}{q}=-\Gamma_{k+1}<0$, which contradicts the condition of $x_i(v_i)>0$. In this case, $q=v_{k+1}$ and $x_{k+1}(v_{k+1})=\sum_{i=k+2}^n \Gamma_i-\sum_{i=1}^k\frac{B_i}{v_{k+1}}$. Due to part (1) of this Lemma, we have 
\begin{eqnarray*}
p_{k+1}(v_{k+1}) &\le& v_{k+1} x_{k+1}(v_{k+1}) = v_{k+1} \left( \sum_{i=k+2}^n \Gamma_i - \sum_{i=1}^k B_i / v_{k+1} \right) \\
&=& v_{k+1} \sum_{i=k+2}^n \Gamma_i - \sum_{i=1}^k B_i < B_{k+1},
\end{eqnarray*}
where the last inequality follows from the definition of $k$.

This lemma holds.

\subsection{Proof of~\cref{lem:boundC}}\label{sec:lem:boundC}
To prove \(\sum_{i \in [n]} -p_i \leq \sum_{i \in [n]} v_i x_i - p_i\), we only need to show that \(\sum_{i \in [n]} v_i x_i \ge 0\).
Since \(\sum_{i \in [n]} x_i = 0\), and the buyers' valuations are always higher than the sellers' valuations, it is evident that \(\sum_{i \in [n]} v_i x_i \ge 0\).

\section{Myerson's Lemma}

\begin{lemma}[Myerson's Lemma~\cite{myerson1981optimal}]\label{lem:Myerson}
The functions $(x_i(v_i, \vecv_{-i}), p_i(v_i, \vecv_{-i}))_{i=1}^n$ satisfy the incentive compatibility (IC) inequality, i.e. $\forall i, \forall v_i, \forall v_i', \forall \vecv_{-i}$, 
$$
v_i x_i(v_i, \vecv_{-i}) - p_i(v_i, \vecv_{-i}) \ge v_i x_i(v_i', \vecv_{-i}) - p_i(v_i', \vecv_{-i}),
$$
if and only if, given valuation profile $\vecv_{-i}$,
\begin{itemize}
    \item[(1)] $x_i(v_i, \vecv_{-i})$ is monotone non-decreasing in $v_i$;
    \item[(2)] $p_i(v_i, \vecv_{-i}) = v_i x_i(v_i, \vecv_{-i}) - \int_0^{v_i} x_i(z, \vecv_{-i}) \, \mathrm{d}z + h_i(\vecv_{-i})$, where $h_i(\vecv_{-i})$ is a term that is independent of $v_i$.
\end{itemize}
\end{lemma}

\section{Chernoff Bound}

Below is a form of the Chernoff Bound used in this paper.

\begin{lemma}[Chernoff Bound]\label{lem:chernoff}
    Let $x_1,...,x_n$ be $n$ independent random variable in $[0,1]$. Denote by $X=\sum_{i\in [n]}x_i$ their sum and by $\mu=\mathbb{E}[X]$ the expectation. For any $0<\delta<1$:
    $$
    \Pr\left[|X-\mu|\ge \delta\mu \right]\leq 2\exp{\left(-\frac{\delta^2\mu}{3}\right)}.
    $$

\end{lemma}

\end{document}